\documentclass{llncs}
\usepackage{makeidx}  
\usepackage[nocompress]{cite}
\usepackage[T1]{fontenc}
\usepackage{amsfonts,amssymb,latexsym,amsmath,enumitem}
\usepackage{setspace}
\usepackage{tikz}
\usepackage{caption,sidecap}
\usepackage{algorithmicx,algorithm}
\usepackage[noend]{algpseudocode}

\newif\ifdraftpaper
\draftpapertrue

\setlist[enumerate]{itemsep=2mm}
\setlist[itemize]{itemsep=2mm}

\DeclareFontFamily{OT1}{pzc}{}
\DeclareFontShape{OT1}{pzc}{m}{it}{<-> s * [1.00] pzcmi7t}{}
\DeclareMathAlphabet{\mathpzc}{OT1}{pzc}{m}{it}
\newcommand{\Lang}[1]{\mathpzc{Lang}(#1)}

\usetikzlibrary{decorations.pathreplacing}

\renewcommand{\phi}{\varphi}

\newcommand{\naw}[1]{{\langle #1\rangle}}
\newcommand{\set}[1]{\{#1\}}
\newcommand{\card}[1]{|#1|}

\newcommand{\ctlks}{CTLK\textsuperscript{*}}
\newcommand{\ehs}{EHS\textsuperscript{+}}
\newcommand{\isrl}{ISRL}
\newcommand{\ehsr}{EHS${}^+_{A\bei{}LN}$}
\newcommand{\ehsd}{EHS${}^+_{BDE}$}
\newcommand{\ehsre}{EHS\textsuperscript{RE}}

\newcommand{\Lehs}{{\cal L}_{\ehs}}
\newcommand{\Lehsre}{{\cal L}_{\ehsre}}
\newcommand{\Var}{\mathit{Var}}
\newcommand{\A}{\mathcal{A}}
\newcommand{\g}{\mathrm{g}}

\newcommand{\ISe}[1]{(\set{L_i, l^0_i, ACT_i, P_i, t_i}_{i
  \in A}, #1)}
\newcommand{\Me}{(S, s_0, t, \set{\sim_i}_{i \in A}, \lambda)}

\newcommand{\PSPACE}{\textsc{PSpace}}

\newcommand{\aei}{\bar{A}}
\newcommand{\bei}{\bar{B}}
\newcommand{\dei}{\bar{D}}
\newcommand{\eei}{\bar{E}}
\newcommand{\oei}{\bar{O}}
\newcommand{\lei}{\bar{L}}
\newcommand{\nei}{\bar{N}}

\newcommand{\modelsb}{\models_B}
\newcommand{\first}{\mathit{first}}
\newcommand{\last}{\mathit{last}}

\ifdraftpaper
   \newcommand{\nb}[1]{$\#$\footnote{#1}}
   
\else
   \newcommand{\nb}[1]{}
      
\fi

\begin{document}
\title{Model Checking Epistemic Halpern-Shoham Logic Extended with Regular Expressions}

\author{Alessio Lomuscio \and Jakub Michaliszyn}
\institute{Imperial College London, UK}

\maketitle

\begin{abstract}
The Epistemic Halpern-Shoham logic (EHS) is a temporal-epistemic logic
that combines the interval operators of the Halpern-Shoham logic with
epistemic modalities. The semantics of EHS is based on interpreted
systems whose labelling function is defined on the endpoints of
intervals. We show that this definition can be generalised by allowing
the labelling function to be based on the whole interval by means of
regular expressions. We prove that all the positive results known for
EHS, notably the attractive complexity of its model checking problem
for some of its fragments, still hold for its generalisation. We also
propose the new logic \ehsre{} which operates on standard Kripke
structures and has expressive power equivalent to that of EHS with
regular expressions. We compare the expressive power of \ehsre{} with
standard temporal logics.
\end{abstract}


\section{Introduction}
Model checking is a leading technique in automatic verification.  The
model checking problem consists of establishing whether a property,
expressed as a logical formula, holds on a system, represented as a
model~\cite{Clarke+99a}. 
Model checking has recently been studied in the context of interval
temporal logic~\cite{LomuscioMichaliszyn13,LomuscioMichaliszyn14c}. In this context temporal
specifications consist of formulas expressing properties of {\em
  intervals} rather than {\em states} as it is traditionally the case
in temporal logic.


Interval temporal logic has a long and succesful tradition in Logic in
Computer Science. The logics ITL~\cite{MO83}, defined by Moszkowski,
and HS~\cite{HS91}, defined by Halpern and Shoham, are the most
commonly used formalisms. ITL suffers from the high-complexity of its
model checking problem which is non-elementary-complete
\cite{Lodaya12}.  In this paper we focus on HS as the basic underlying
framework.
HS is a modal temporal logic in which the elements of a model are
pairs of points in time, or {\em intervals}. For an interval $[p,q]$
it is assumed that $q$ happens no earlier than $p$, but no assumption
is made on the underlying order, which can be discrete, continuous,
linear, branching, etc.

Traditionally, twelve modal operators acting on intervals are defined
in HS. They are: $A$ (``after/meets''), $B$ (``begins''), $D$
(``during''), $E$ (``ends''), $L$ (``later''), $O$ (``overlaps'') and
their duals: $\bar{A}$, $\bar{B}, \bar{D}, \bar{E}, \bar{L}$,
$\bar{O}$.  Some of them are redundant; for example, $B$ and
$E$ can define $D$ (a prefix of a suffix is an
infix)~\cite{Della+13a,Monica+11a}.

The analysis of HS and its fragments is traditionally limited to its
satisfiability problem. This is known to be undecidable in
general~\cite{HS91,BMMSS12,M11}, even when HS is restricted to its
unimodal fragments~\cite{Bresolin+11a}. 
Notable decidable fragments are the $A\bar{A}$ fragment with length constraints~\cite{BDGMS10}, the $AB\bar{B}\bar{L}$  fragment \cite{MPS10}, and the recently
identified Horn fragment~\cite{Artale+15a}. 
Some fragments are decidable only over some particular classes of
orderings.  For example, the $B\bar{B}D\bar{D}L\bar{L}$ fragment was
shown to be decidable over the class of all dense orders~\cite{MPS09},
while the $D$ fragment is undecidable over discrete
orders~\cite{MM11}. The same logic becomes decidable if one assume
that an interval is its own infix \cite{APS10}. While a wealth of
results have been put forward, open questions remain. For example, the
decidability of the $D$ fragments over the class of all orders is
currently open.

\textbf{The logic EHS.}  In applications, temporal logics often appear
in combination with other modalities expressing other aspects of the
system or its components. A notable example is temporal-epistemic
logic~\cite{Fagin+95b} where the knowledge of the components, or {\em
  agents}, is assessed from an information-theoretic point of view.
Temporal-epistemic logic is widely explored in applications, including
security; dedicated model checkers have been
released~\cite{mck-tool,verics-tool,LomuscioQuRaimondi09}.

In the traditional approach, the underlying temporal logic is
state-based, either in its linear or branching variants.  A notable
exception to this is the Epistemic HS logic
(EHS)~\cite{LomuscioMichaliszyn13}, which consists of a combination of
epistemic modalities with the interval-based temporal logic HS. EHS
combines all the HS interval-temporal modalities with standard
epistemic modalities: $K_i$ (``agent $i$ knows that'') and $C_\Gamma$
(``it is a common knowledge in group of agents $\Gamma$ that''). The
logic EIT, a simple fragment of EHS where only epistemic modalities
are allowed, but modalities are interpreted on intervals rather than
points, has been shown to be \PSPACE{}-hard.  Model checking of the
$BDE$-fragment of EHS with epistemic operators is \PSPACE{}-complete.
Finally, in \cite{LomuscioMichaliszyn14c} it was shown that the
$A\bar{B}L$ fragment of EHS has a decidable model checking problem.

The labelling function in the structures considered in
\cite{LomuscioMichaliszyn13} is defined on the endpoints of the
intervals.  This corresponds to the intuitive representation of
intervals as pairs
and is often adopted in the literature. However, other choices are
possible. For example,~\cite{Montanari+14a} considers the labelling
for an interval as the intersection of the labelling of all its
elements. We argue that even more expressive setups are required. 

Assume, for example, that we need to label a whole
process of printing by means of the propositional variable $printing$.
By adopting \cite{Montanari+14a}, by labelling the process with
$printing$, it would follow that every subinterval would need to be
labelled with $printing$ too. This may not correspond
to our intuition.

%
Similarly, if we were to adopt a labelling based on endpoints, and $S$
($E$) is the state where printing starts (ends, respectively), it
would follow that all the intervals starting in $S$ and ending in $E$
have to be labelled with $printing$.  But if more than one process is
present, it follows that the interval starting at the beginning of the
first process and ending at the end of the second one is also labelled
with $printing$, which, again, may be against our intuition.

This is just a simple example (we explore more significant ones in
Section~\ref{sec:examples}); but it suggests that more liberal labellings imposing no
such constraints are called for in this context. From a theoretical
standpoint, it is of interest to generalise previous labelling
approaches and assess the impact these have on the decidability of the
model checking problem. We are not aware of any previous attempt in
this direction in the context of any HS logic.

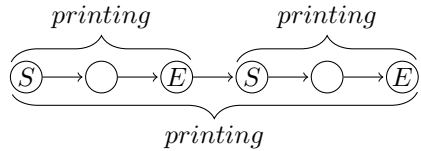
\begin{center}
 \begin{tikzpicture}[shorten >=1pt,->, scale=1]
  \tikzstyle{vertex}=[circle,draw=black,minimum size=12pt,inner sep=0pt]

\begin{scope}[shift={(0,0)}]
  \foreach \var/\x/\y/\name in {
    S/0/1/s1, 
    ~/1/1/s2, 
    E/2/1/s3,
    S/3/1/s4, 
    ~/4/1/s5, 
    E/5/1/s6}
    \node[vertex] (G-\name) at (\x,\y) {$\var$};

  \foreach \x/\y in {
    1/2,
    2/3,
    3/4,
    4/5,
    5/6}
   \draw (G-s\x) to (G-s\y);

\draw[decorate,decoration={brace,amplitude=10},-] (-0.2,1.2) -- (2.2,1.2) ;
\node (X) at (1.0,1.8) {$printing$};

\draw[decorate,decoration={brace,amplitude=10},-] (2.8,1.2) -- (5.2,1.2) ;
\node (X) at (4,1.8) {$printing$};

\draw[decorate,decoration={brace,amplitude=10},-] (5.2,0.8) -- (-0.2,0.8) ;
\node (X) at (2.5,0.2) {$printing$};

\end{scope}
\end{tikzpicture}
\end{center}

{\bf Contribution.}  We put forward a generalisation of the labelling
functions independently proposed in \cite{LomuscioMichaliszyn13} and
\cite{Montanari+14a}. The novel labelling is defined by using regular
expressions based on the states of the whole interval.  For example,
the process of printing from the example above can now be modelled by
using the regular expression $S{\neg {E}}^*E$.  The models that result
from this labelling are here called interpreted systems with regular
labelling, \isrl{} for short. We study the logic \ehs{}, sharing the
syntax of EHS, but interpreted over \isrl{}, and show that it enjoys
all the positive results known for EHS.

In order to be able to express properties of standard point-based
models, and formally characterise the expressive power of \ehs{}, we also
define and study the logic \ehsre{}. Intuitively, \ehsre{} can be seen
as the result of moving the regular expressions from the labelling
function to the atomic propositions.  We show polynomial time
reductions between the model checking problems for \ehsre{} and \ehs{}
and characterise the expressive power of the former.

\textbf{Related work.}
Initial results for the model checking of HS and some of its variants
have appeared
recently~\cite{LomuscioMichaliszyn13,LomuscioMichaliszyn14c,Montanari+14a}.
The results of this paper generalise those presented
in~\cite{LomuscioMichaliszyn13,LomuscioMichaliszyn14c}.  Our setting
is more expressive than~\cite{LomuscioMichaliszyn13} and further
benefits from the fact that many properties become easier to express.

Note that ITL does allow for regular expressions to be used. Unlike
\ehsre{}, where regular expressions can be used only for propositions,
in ITL they can be used for any subformula. However, ITL expresses
properties of a single interval, while \ehsre{} can express properties
of different branches.  Furthermore, HS enjoys several fragments, such
as the $BDE$ one, with a computationally attractive model checking
problem. This may be of particular use in applications.

Two further formalisms that are related to \ehsre{} are PDL
\cite{HarelTiurynKozen00a} and its linear counterpart LDL
\cite{DeGiacomoVardi13a}.  An epistemic version of PDL, E-PDL, was
proposed in \cite{BenthemEijckKooi06a}.  However, epistemic modalities
in E-PDL are interpreted on points, not intervals as in EHS and
\ehsre{}. This is largely the reason why \ehsre{} is more expressive
than E-PDL and the model checking problem for E-PDL is decidable in
polynomial time~\cite{Lange06}, whereas the model checking problem for
EIT is already \PSPACE{}-hard.  Notice also that E-PDL does not have
backward modalities and can express properties of actions, unlike
\ehsre{}.

Results on the correspondence between regular expressions and HS were
presented in~\cite{MontanariSala13a}, where it was shown that each
$\omega$-regular language can be encoded in the $AB\bar{B}$ fragment
of HS.  The encoding, however, uses additional propositional variables
to label interval, and therefore cannot be used for the model checking
problem.
%
%

\section{Interpreted systems with regular labelling}
We begin by recalling the notions of regular expressions.  Given a set
$X$, the set of regular expressions over $X$, denoted by $RE_X$, is
defined by the following BNF:

\[e ::= \emptyset \mid \epsilon \mid s \mid e;e \mid e+e \mid e^*\]
where $s\in X$. We allow parentheses for grouping and often omit the concatenation
symbol ``$;$''.

For each regular expression $e$, let $\Lang{e}$ stand for the language \emph{denoted} by $e$. Formally,
	$\Lang{\emptyset} = \emptyset$,
	$\Lang{\epsilon} = \set{\epsilon}$,
	$\Lang{s} = \set{s}$,
	$\Lang{e_1;e_2} = \set{w_1w_2 \mid w_1 \in \Lang{e_1} \land w_2 \in \Lang{e_2}}$,
	$\Lang{e_1+e_2} = \Lang{e_1} \cup \Lang{e_2}$, and
	$\Lang{e^*}$ is the smallest set containing $\epsilon$ such that for all $w_1 \in L(e)$ and $w_2 \in \Lang{e^*}$, $w_1w_2 \in \Lang{e^*}$.

Now we generalise interval-based interpreted
systems~\cite{LomuscioMichaliszyn13} to systems with labelling based
on regular expressions.

\begin{definition}
Given a set of agents $A=\{0, 1, \dots, m\}$, an \emph{interpreted
  system with labelling on regular expressions}, ISRL for short, is a
tuple $IS=\ISe{\lambda}$, where:
\begin{itemize}
\setlength\itemsep{0.3em}
\item $L_i$ is a finite set of local states for agent $i$,
\item $l^0_i \in L_i$ is the initial state for agent $i$,
\item $ACT_i$ is a finite set of local actions available to agent $i$,
\item $P_i:L_i \rightarrow 2^{ACT_i}$ is a local protocol function for
  agent $i$, returning the set of possible local actions in a given
  local state,
\item $t_i\subseteq L_i \times ACT \times L_i$, where $ACT = ACT_0
  \times \dots \times ACT_m$, is a local transition relation returning
  the next local state when a joint action is performed by all agents
  on a given local state,
\item $\lambda: \Var{} \to RE_G$ is a labelling function,
  where $G= L_0 \times L_1 \times \dots \times L_m$ is the set of \emph{global configurations} and $\Var{}$ is a finite
  set of propositional variables.
\end{itemize}
\end{definition}
Agent $0$ is often called \emph{the environment}.

We now define models of an IS on sets of paths from its initial configuration.
Let $t^G \subseteq G^2$ be a relation such that
$t^G((l_0, \dots, l_m), (l'_0, \dots, l'_m))$ iff there exists a joint
action $(a_0, \dots, a_m) \in ACT$ such that for all $i$ we have $a_i
\in P_i(l_i)$ and $t_i(l_i, (a_0, \dots, a_m), l_i')$.  

\begin{definition}\label{def:ibis}
Given an \isrl{} 
$IS=\ISe{\lambda}$ over a set of agents $A=\{0, \dots, m\}$, \emph{the model of the $IS$} is a tuple
$M=\Me{}$, where
\begin{itemize}
\setlength\itemsep{0.3em}
\item $S\subseteq G^+$ is the set of global states, i.e.,
  non-empty sequences $g_0\dots g_k$ such that $g_0 =(l^0_0, \dots, l^0_m)$ and for each $i < k$ we have $t^G(g_i, g_{i+1})$,
\item $s_0 = g_0 = (l^0_0, \dots, l^0_m)$ is the initial state of the system,
\item $t \subseteq S^2$ is the global transition relation such that $t(g_0\dots g_k, g'_0 \dots g_l')$ iff $l=k+1$ and for all $i \leq k$ we have $g_i = g_i'$,
\item $\sim_i \subseteq S^2$ is the equivalence relation such that $g_0\dots g_k \sim_i g'_0 \dots g_l'$ iff $g_k=(l_0, \dots, l_m)$, $g'_l=(l'_0, \dots, l'_m)$ and $l_i = l'_i$, and
\item $\lambda$ is the labelling function.
\end{itemize}
\end{definition}

Intuitively, $S$ denotes the set of global configurations of the \isrl{}
equipped with information about all their predecessors. This is the
standard construction used for defining unravelling
in temporal logic (see, e.g., Definition 4.51 in \cite{Blackburn+01a}). We need to keep the information regarding the predecessors for the semantics of backward modalities; the semantics of the epistemic modalities is defined only on the current state.

Given a model $M$, an \emph{interval} in $M$ is a finite path on $M$,
i.e., a sequence of states $I=s_1, s_2, \dots, s_n$ such that $t(s_i,
s_{i+1})$, for $1 \leq i\leq (n-1)$. A \emph{point interval} is an interval
that consists of exactly one state.  We assume $pi(I)=\top$ for a point
interval $I$ and $pi(I)=\bot$ for all the other intervals.

For each state of $s= g_0, \dots, g_k \in S$, we assume $\g(s)=g_k$.  So
$\g(s)$ denotes the actual states of $s$, not its history.  We extend
$\g$ to intervals by assuming $\g(I)=\g(s_0)\dots \g(s_k)$ for every
interval $I=s_0, \dots, s_k$.

We say that an \isrl{} is \emph{point-based} if $\lambda$ only labels the
point intervals, i.e., for each $v \in \Var{}$ we have
\(\lambda(v) = \sum_{s \in S'} s\) for some $S' \subseteq S$. 
An \isrl{} is \emph{endpoint-based} if $\lambda$ is defined on the endpoints of the intervals, i.e., for
each $v \in \Var{}$ we have
\(\lambda(v) = \sum_{s \in S'} (s + sS^*s) + \sum_{(s, s') \in P} sS^*s'\)
for some $S' \subseteq S$, $P \subseteq S^2 \setminus \set{(s,s) \mid s \in S}$.
Notice that the models of the point-based \isrl{} can be seen as
standard Kripke structures; the models of the endpoint-based \isrl{}
can be seen as the generalised Kripke structures of
\cite{LomuscioMichaliszyn13}.

For $g=(l_0, l_1, \dots, l_m)$ we denote by $l_i(g)$ the local state
$l_i \in L_i$ of agent $i \in A$ in $g$. For a global state $s=g_0,
\dots, g_k$, we assume $l_i(s)=l_i(g_k)$.

Now we give an example of an interpreted system and of its model.  We
will use this example in the following sections to illustrate other
constructions.

\begin{example}\label{e:re1}
  Consider an
\isrl{}
$IS_{ex} = \ISe{\lambda})$
over a set of agents $A=\set{0,1}$ and a set of propositional variables $Var=\set{p}$, where 
\begin{itemize}
\setlength\itemsep{0.3em}
\item $L_0=\set{l_0}$, $L_1=\set{l_1, l_2, l_3}$,
\item $l^0_0=l_0$, $l^0_1=l_1$,
\item $ACT_0=\set{a_1, a_2}$, $ACT_1=\set{\epsilon}$,
\item $P_0(l_0)=ACT_0$, $P_1(l_1)=P_1(l_2)=P_1(l_3)=ACT_1$,
\item $t_0=\set{(l_0, (a_1,\epsilon), l_0),(l_0, (a_2,\epsilon), l_0)}$, 
 $t_1=\set{(l_1, (a_1,\epsilon), l_2)$, 
$(l_1, (a_2,\epsilon), l_2)$, \\
$(l_2, (a_2,\epsilon), l_3)$, 
$(l_2, (a_1,\epsilon), l_1)$, 
$(l_3, (a_1,\epsilon), l_1)$, 
$(l_3, (a_2,\epsilon), l_1)}$, 
\item $\lambda(p) = g_1 (g_1+g_2)^* g_3$, where $g_i=(l_0, l_i)$.
\end{itemize}
Figure \ref{f:exagents} depicts the agents of $IS$.
We have $G=\set{g_1,g_2,g_3}$
and 
$t^G = \set{ 
((l_0, l_1),\allowbreak(l_0, l_2)),\allowbreak
((l_0, l_2),\allowbreak(l_0, l_3)),\allowbreak
((l_0, l_2),\allowbreak(l_0, l_1)),\allowbreak
((l_0, l_3),(l_0, l_1))
}$.
The model $M_{ex}$ of $IS_{ex}$ is infinite. Its fragment is depicted
in Figure \ref{f:exmodel}. 
\end{example}

\begin{figure}[!t]
\centering \begin{tikzpicture}[shorten >=1pt,->, scale=1]
  \tikzstyle{vertex}=[circle,draw=black,minimum size=12pt,inner sep=0pt]

\begin{scope}[shift={(1,0)}]
  \foreach \var/\x/\y/\name in {l_1/0/1/s1, l_2/2/1/s2, l_3/4/1/s3}
    \node[vertex] (G-\name) at (\x,\y) {$\var$};

  \draw (G-s1) to node[below] {$(*, \epsilon)$} (G-s2);    
  \draw (G-s2) to node[below] {$(a_2, \epsilon)$} (G-s3);    
    
  \draw[out=165,in=8,->] (G-s3) to node[above] {$(*, \epsilon)$}  (G-s1);
  \draw[out=22,in=158,<-] (G-s1)  to node[above] {$(a_1, \epsilon)$} (G-s2);
\end{scope}

\begin{scope}[shift={(-2,0.0)}]
  \foreach \var/\x/\y/\name in {l_0/0/1/s1}
    \node[vertex] (G-\name) at (\x,\y) {$\var$};
  
  \draw[loop right] (G-s1) to node[right] {$(*, \epsilon)$}  (G-s1);    
\end{scope}

\end{tikzpicture}
\caption{The agents from Example \ref{e:re1}, where $*$ stands for any action.}\label{f:exagents}
\end{figure}
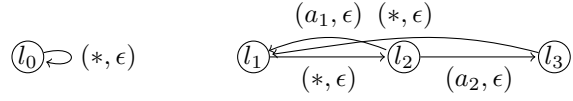

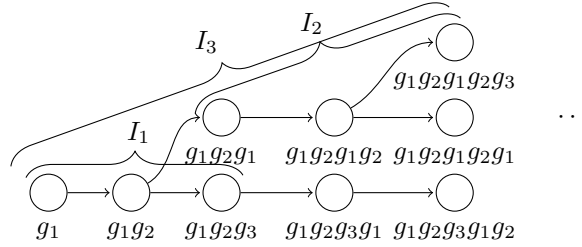
\begin{figure}[!t]
\centering
\usetikzlibrary{decorations.pathreplacing}
\begin{tikzpicture}[shorten >=1pt,->, scale=1]
  \tikzstyle{vertex}=[circle,draw=black,minimum size=14pt,inner sep=0pt]

\begin{scope}[shift={(0,0)}]
  \foreach \var/\x/\y/\name in {g_1/2/2/12, g_1g_2/3.1/2/22,
g_1g_2g_3/4.3/2/32,
g_1g_2g_3g_1/5.8/2/42,
g_1g_2g_3g_1g_2/7.4/2/52,
g_1g_2g_1/4.3/3.0/33,
g_1g_2g_1g_2/5.8/3.0/43,
g_1g_2g_1g_2g_1/7.4/3.0/53, 
g_1g_2g_1g_2g_3/7.4/4.0/54}, 
    \node[vertex,label={[label distance=-1.0cm]90:$\var$}] (G-\name) at (\x,\y) {~};

  \foreach \from/\to in
{12/22,22/32,32/42,42/52,33/43,43/53}
    \draw (G-\from) -- (G-\to);
    
  \draw[out=30,in=180,->] (G-43) to (G-54);
  \draw[out=30,in=180,->] (G-22) to (G-33);

\coordinate [label=center:$\cdots$] (A) at (9, 3);

\draw[decorate,decoration={brace,amplitude=10},-] (1.7,2.2) -- (4.6,2.2) ;
\node (X) at (3.2,2.8) {$I_1$};

\draw[decorate,decoration={brace,amplitude=10},-] (3.95,3.051) -- (7.5,4.30) ;
\node (X) at (5.5,4.25) {$I_2$};

\draw[decorate,decoration={brace,amplitude=10},-] (1.5,2.35) -- (7.35,4.45) ;
\node (X) at (4.1,4.0) {$I_3$};

\end{scope}	

\end{tikzpicture}
\caption{A fragment of the model of $IS_{ex}$ from Example \ref{e:re1}. $I_1$, $I_2$ and $I_3$ are labelled by $p$, as $\g(I_1)=\g(I_2)=g_1g_2g_3$ and $\g(I_3)=g_1g_2g_1g_2g_3$ belong to $\Lang{\lambda(p)}$. }\label{f:exmodel}
\end{figure}

\section{The logic \ehs{}}
\begin{figure}[!t]
\begin{tikzpicture}[scale=1]
\draw[|-|,semithick] (0,6.3) -- (2,6.3);
\draw[|-|,semithick] (1.97,6) -- (3.1,6);
\draw[|-|,semithick] (0,5.6) -- (1.3,5.6);
\draw[|-|,semithick] (0.3,5.2) -- (1.7,5.2);
\draw[|-|,semithick] (0.6,4.8) -- (2,4.8);
\draw[|-|,semithick] (2.5,4.4) -- (3.1,4.4);
\draw[|-|,semithick] (1,4) -- (3.1,4);

\draw[-,dotted] (0.015,6.3) -- (0.015,3.7);
\draw[-,dotted] (1.985,6.3) -- (1.985,3.7);

\node[text width=8.5cm] (note2) at (-4.2,6) {$I R_A I'$ iff $first(I')=last(I)$};
\node[text width=8.5cm] (note2) at (-4.2,5.6) {$I R_B I'$ iff $I=I'I_1$ 
 for some interval $I_1$ };
\node[text width=8.5cm] (note2) at (-4.2,5.2) {$I R_D I'$
 iff $I = I_1I' I_2$
 for some intervals $I_1,I_2$};
\node[text width=8.5cm] (note2) at (-4.2,4.8) {$I R_E I'$ iff $I=I_1I'$
 for some interval $I_1$ };
\node[text width=8.5cm] (note2) at (-4.2,4.4) {$I R_L I'$ iff there is a path from
  $last(I)$ to $first(I')$};
\node[text width=8.5cm] (note2) at (-4.2,4) {$I R_O I'$ iff $II_1=I_2I'$
 for some intervals $I_1, I_2$};
\end{tikzpicture}
\caption{Basic Allen relations.}\label{f:allens}
\end{figure}
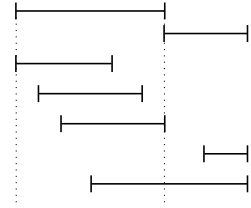

We now define the syntax of the specification language we focus on in
this paper. The temporal operators represent 
relations between intervals as originally defined by Allen
\cite{Al83}.  Six of these relations are presented in
Figure~\ref{f:allens}: $R_A$ (``{\bf a}fter'' or
``meets''), $R_B$ (``{\bf b}egins'' or ``starts''), $R_D$ (``{\bf
  d}uring''), $R_E$ (``{\bf e}nds''), $R_L$ (``{\bf l}ater''), and
$R_O$ (''{\bf o}verlaps''). Six additional operators can be defined
corresponding to the six inverse relations. Formally, for each $X
\in \{A,B,D,E,L,O\}$, we also consider the relation $R_{\bar X}$,
corresponding to ${R_X}^{-1}$. 

For convenience, we also consider the ``next'' relation $R_N$ such that $I R_N
I'$ iff $t(last(I), first(I'))$ \cite{LomuscioMichaliszyn14c}.  Let $\mathbb{HS} =\set{A,\aei{}, B, \bei{}, D, \dei{}, E, \eei{},
  L, \lei{}, N, \nei{}, O, \oei{}}$.


\begin{definition}\label{def:ehs} The syntax of the Epistemic
  Halpern--Shoham Logic (\ehs{}), $\Lehs{}$ is defined by the following BNF.
\[\begin{array}{rcl}
\varphi &::=& pi \mid p \mid  \neg \varphi \mid  \varphi \wedge \varphi \mid  K_i \varphi \mid  C_\Gamma \varphi \mid \langle X \rangle \varphi 
\end{array}
\]
where $p\in \Var{}$ is a propositional variable, $i\in A$ is an agent,  $\Gamma\subseteq A$ is a set of agents, and $X \in \mathbb{HS}$. 
\end{definition}

We use 
abbreviations including
$[X]\varphi$ for $\neg \naw{X} \neg \varphi$ and the usual Boolean
connectives $\vee$, $\Rightarrow$, $\Leftrightarrow$ as well as the
constants $\top, \bot$ in the standard way. 

Note that the modality $\naw{N}$ is a counterpart of the $EX$ operator
of CTL.  While $\naw{N}$ is redundant in \ehs{} since $\naw{N} \varphi
= \naw{A} (\neg pi \wedge \naw{B}\naw{B}\bot \wedge \naw{A} \varphi)$,
it is useful in fragments of \ehs{} that do not contain $B$ and $E$.

In order to provide the semantics for the epistemic operators on an
interval based semantics, we specify when two intervals are
epistemically indistinguishable for an agent, i.e., an agent cannot
distinguish between the two. We say that $I\sim_i I'$, where $I=s_1,
\dots, s_k$, $I'=s'_1, \dots, s'_l$, iff $k=l$ and for all $j\leq k$ we
have $s_j \sim_i {s'}_j$. In other words, for two intervals to be
indistinguishable to agent $i$ the two intervals need to be of the
same length and the agent cannot be able to distinguish any
corresponding point in the interval. This appears the natural
generalisation to intervals of the point-based knowledge modalities
traditionally used in epistemic logic~\cite{Fagin+95b}.
 For example, in the model presented in Example \ref{e:re1}, we have $I \sim_0
I'$ if and only if $|I|=|I'|$ and $I\sim_1 I'$ if and only if $I=I'$;
in general these relations may be more complicated.
We extend this definition to the common knowledge case by considering
$\sim_\Gamma = (\bigcup_{i \in \Gamma} \sim_i)^*$, for any group of agents
$\Gamma\subset A$,
where ${}^*$ denotes the transitive closure.
 For further
explanations we refer
to~\cite{LomuscioMichaliszyn13}.


We now define when a formula is satisfied in an interval on an \isrl{}.

\begin{definition}[Satisfaction]\label{def:satisfaction}
Given an \ehs{} formula $\varphi$, an \isrl{} $IS$, its model $M=\Me{}$ and an interval
$I$, we inductively define whether $\varphi$ holds in the interval
$I$, denoted $M, I \models \varphi$, as follows:
\begin{enumerate}[label=(\roman*)]
\setlength\itemsep{0.3em}
\item $M, I  \models pi$ iff $I$ is a point interval,
\item\label{i:proposition} $M, I  \models p$ iff $\g(I) \in \Lang{\lambda(p)}$,
\item $M, I  \models \neg \varphi$ iff it is not the case that $M, I
  \models \varphi$,
\item $M, I \models \varphi_1 \wedge \varphi_2$ iff $M, I \models
  \varphi_1$ and $M, I \models \varphi_2$,
\item $M, I \models K_i \varphi$, where $i \in A$, iff for all
  $I{}' \sim_i I $ we have $M, I{}' \models \varphi$,
\item $M, I \models C_\Gamma \varphi$, where $\Gamma\subseteq A$, iff for all $I{}' \sim_\Gamma I $ we have $M,
  I' \models \varphi$,
\item $M, I \models \naw{X} \varphi$ iff there exists an interval
  $I{}'$ such that $I R_X I{}'$ and $M, I{}' \models \varphi$,
  where $R_X$ is an Allen relation as above.
\end{enumerate}
\end{definition}

We write $IS, I \models \varphi$ if $M, I \models \varphi$, where
$M$ is the model of $IS$, and $IS \models \varphi$ if $IS, s_0 \models \varphi$.

\vspace{-1em}
\section{Expressive power} \label{sec:examples}
\vspace{-1em}
The
expressivity of \ehs{} is incomparable to that of traditional
formalisms such as LTL, CTL, or EHS as \ehs{} is defined on different semantics structures.  To investigate its
expressive power, we introduce \ehsre{}, a variant of \ehs{}
defined over point-based interpreted systems.
We show that the model checking problems for \ehsre{} and \ehs{} admit
a polynomial time reduction to one another on the corresponding
semantics. We also observe that \ehsre{} can represent properties not
expressible by \ctlks{}, the epistemic version of CTL$^*$ (and therefore LTLK and CTLK). So,
intuitively, there is a sense in which \ehs{} is indeed more
expressive than the usual temporal-epistemic logic interpreted on points.

For a labelling function $\lambda$ and a regular expression $r$, let $\lambda \circ r$ be the regular expression obtained from $r$ by replacing each propositional variable $p$ by $\sum_{g \in \lambda(p)} g$ (if $\lambda(p)=\emptyset$, we put $\emptyset$).

\begin{definition}\label{def:ehsre} The language of \ehsre{}, $\Lehsre{}$, is defined as follows:
\[\begin{array}{rcl}
\varphi &::=& pi \mid r \mid \neg \varphi \mid \varphi \wedge \varphi
\mid K_i \varphi \mid C_\Gamma \varphi \mid \langle X \rangle \varphi
\end{array}
\]
where $r\in RE_{2^\mathit{Var}}$, $i\in A$, $\Gamma\subseteq A$, and
$X \in \mathbb{HS}$.

The semantics of \ehsre{} results from replacing the second rule in Definition \ref{def:satisfaction} by (ii') $M, I \models r$ iff $I=s_1, \dots, s_k$ and
  $\g(s_1) \dots \g(s_k) \in \Lang{\lambda \circ r}$.
\end{definition}

Intuitively, \ehsre{} is the result of adapting \ehs{} by moving the
regular expressions from the labelling function into the language.

For convenience, we allow to use $p$ and $\neg p$ in the regular
expressions, by defining
\(p=\sum_{X \subseteq \Var, p \in X} X\)
and
\(\neg p=\sum_{X \subseteq \Var, p \not \in X} X\).

Let $\mathbb{L}_\mathit{Var}$ be the set of all the possible
labellings of interpreted systems with variables of $\mathit{Var}$,
and $\mathbb{L}^{pi}_\mathit{Var} \subset \mathbb{L}_\mathit{Var}$ be
the set of all such labellings for point-based interpreted systems.

\begin{theorem}\label{t:reductions}
There exist polynomial time computable functions 
$f: \mathbb{L}_\mathit{Var}  \times \Lehs \to \mathbb{L}^{pi}_\mathit{Var} \times \Lehsre$ and 
$f': \mathbb{L}^{pi}_\mathit{Var} \times \Lehsre \to \mathbb{L}_\mathit{Var}  \times \Lehs $ 
such that for any interpreted system $IS=\ISe{L}$, any formula $\varphi$ and any interval $I$:
\begin{enumerate} 
\item If $IS, I \models \varphi$ and $f(L, \varphi)=(L', \varphi')$, then 
$IS'=\ISe{L'}$ is point-based and such that $IS', I \models \varphi'$.
\item If IS is point-based, $IS, I \models \varphi$, and $f'(L, \varphi)=(L', \varphi')$, then we have that
$\ISe{L'}, I \models \varphi'$.
\end{enumerate}
\end{theorem}

Given Theorem \ref{t:reductions}, we can say that the logics \ehs{}
and \ehsre{} can describe the same properties of corresponding
interpreted systems.  Since \ehsre{} expresses properties of
point-based interpreted systems, whose models are standard Kripke
structures, we can formally compare the expressive power of \ehsre{}
to that of some more widely known formalisms.

\begin{definition}
Given two logics $\mathcal{L}_1, \mathcal{L}_2$, we write
$\mathcal{L}_1 \subseteq \mathcal{L}_2$ if for each formula
$\varphi_1$ of $\mathcal{L}_1$ there is a formula $\varphi_2$ of
$\mathcal{L}_2$ such that for all point-based \isrl{} we have $IS
\models \varphi_1$ iff $IS \models \varphi_2$.
\end{definition}

One can easily show that \ehsre{} $\not \subseteq$ \ctlks{}.  Consider
the temporal property ``all the paths starting in the initial state
satisfy $(p;True)^\omega$''.  This property cannot be expressed in
\ctlks{}~\cite{Wolper83a}.  However, the property can be verified by
evaluating the \ehsre{} formula
\(p \wedge [A]( (p;\top)^* \Rightarrow [N] (p;\top^*)\).

Also observe that the property above cannot be expressed in the logic
EHS considered over point-based \isrl{} either. So
over point-based \isrl{} we have that  \ehsre{} $\not \subseteq$ EHS

In terms of limitations, note that \ehsre{} can only express
properties of finite intervals.  For example, the CTL property $AFp$
expressing the fact that each infinite path satisfies $p$ at some
point cannot be encoded by any \ehsre{} formula. Therefore CTLK $\not
\subseteq$ \ehsre; similarly we have LTLK $\not \subseteq$ \ehsre.

Since \ehsre{} does not allow us to name actions explicitly, we have
that E-PDL $\not \subseteq$ \ehsre{}.  It can also be shown that
\ehsre{} $\not \subseteq$ E-PDL, since E-PDL cannot express the property
\( \naw{A} (K_1 (pq^*r))\)
as the epistemic modalities in E-PDF is based on states rather than
time-intervals.

\vspace{-1em}
\section{The model checking problem}
\label{sec:decidability}
\vspace{-1em}
We now investigate the complexity of the model checking problem for
fragments of the logics explored so far.

\begin{definition}
Given a formula $\varphi$ of a logic $L$, an \isrl{} $IS$ and an interval $I$, \emph{the model checking problem for
  $L$} amounts to checking whether or not $IS, I \models \varphi$.
\end{definition}

In establishing the above, we say we have model checked $M$ against
the specification $\varphi$ at an interval $I$. Notice that the formula is verified only in the given interval; however, one can easily check whether \emph{all} the initial intervals satisfy a formula $\varphi$ by checking whether  $M, s_0 \models [A] \varphi$. 

The $A\bei{}LN$ fragment of \ehs{}, denoted as \ehsr{}, is
the subset of \ehs{} where the BNF is restricted to the only
modalities $K_i$, $C_\Gamma$, $\naw{A}$, $\naw{\bei}$, $\naw{L}$ and
$\naw{N}$.
Similarly, the $BDE$ fragment of \ehs{}, denoted as \ehsd{}, is the
restriction of \ehs{} to the modalities $K_i$, $C_\Gamma$, $\naw{B}$,
$\naw{D}$ and $\naw{E}$.

\begin{theorem}\label{t:bde}
Model checking \isrl{} against \ehsd{} specifications is decidable
and \PSPACE{}-complete.
\end{theorem}

The above follows from the fact that the satisfaction can be determined by examining only intervals of bounded length. The proof is in the appendix.

\begin{theorem}\label{t:abldec}
Model checking \isrl{} against \ehsr{} specifications is
decidable in non-elementary time.
\end{theorem}

We prove this by generalising the proof of Theorem~13 given
in~\cite{LomuscioMichaliszyn14c}. 

A \emph{top-level} sub-formula of a formula $\varphi$ is a sub-formula of $\varphi$ of the form $X \varphi'$, for some modality $X$ of \ehsr{}, that is not in scope of any modality.
Assume an \isrl{} $IS$. Let $f^{IS}(\varphi)$ be defined recursively as 
   \[f^{IS}(\varphi) = (2|G|^2 \prod_{q \in \Var} 2^{\card{\lambda(q)}}) \cdot 2^{f^{IS}(\varphi_1)} \cdot \ldots \cdot 2^{f^{IS}(\varphi_k)}\] 
   where $X_1\varphi_1$, \dots, $X_k\varphi_k$ are the top-level sub-formulas of $\varphi$. The idea is that $f^{IS}(\varphi)$ is an upper bound on the number of  different \emph{interval types} w.r.t. $\varphi$; an interval type consists of an information whether an interval is a point interval or not (hence $2$), what are its endpoints (hence $|G|^2$), what are the states of the automata corresponding to the regular expressions after reading the interval (hence the product) and types of intervals related to the interval w.r.t. the top level sub-formulas of $\varphi$ (hence the recursive part). 

We define a bounded satisfaction relation $\modelsb$ for \ehsr{}, for which the decidability of the model checking is straightforward. The rules (i'-vi') of the definition of $\modelsb$ are the same as the rules (i-vi) from Definition \ref{def:satisfaction} except that $\models$ is replaced with $\modelsb$. The last rule, however, is different:
\begin{enumerate}
\item[(vii')] $M, I  \modelsb \naw{X} \varphi$ if and only if there exists an interval
  $I{}'$ such that $|I'| \leq |I|+f^{IS}(\varphi)$,  $I R_X I{}'$ and $M, I{}' \modelsb \varphi$,  where $X$ is $A$, $\bar{B}$, or $N$.
\end{enumerate}

It is not hard to see that model checking is decidable for the bounded semantics. It turns out that in the \ehsr{} case, the relations $\models$ and $\modelsb$ are the same, and therefore the model checking procedure for the bounded semantics solves the model checking problem for the unbounded semantics. All the details are in the appendix.

By employing the polynomial time reductions of Theorem
\ref{t:reductions}, we can show that model checking point-based \isrl{} against $BDE$ fragment of \ehsre{}
  specifications is \PSPACE{}-complete and that 
model checking point-based \isrl{} against $A\bei{}LN$ fragment of \ehsre{}
specifications
is decidable.

\vspace{-1em}
\section{Conclusions and Future Work}\label{sec:conclusions}
\vspace{-1em}

Temporal logic is one of the key foundational tools to reason about
computing systems. Several variants of temporal logics have been
studied, reflecting the underlying assumptions on the temporal flow,
ranging from linear to branching and from discrete to continuous.
Interval temporal logics~\cite{MO83,HS91} are a relatively less
explored variant of temporal logic. As is known, these are
particularly appropriate to study the properties of continuous
processes. However, while interval temporal logics could provide a
formal basis for systems verification, little is known in terms of
their model checking problem. Indeed, this was only recently explored
in~\cite{LomuscioMichaliszyn13,LomuscioMichaliszyn14c,Montanari+14a}
in the context of variants of the logic HS.

Since the complexity of the model checking problem for HS fragments is
typically high and the decidability of the full HS logic is not known,
a compelling avenue of research involves establishing whether the
expressivity of previously studied, well-behaved fragments of HS can
be significantly enriched without losing the attractiveness of their
model checking problem.
The logic~\ehs{}, proposed in this paper, combines the interval
temporal logic HS and epistemic logic.  The logic can be see as a
considerable generalisation of the logics proposed
in~\cite{LomuscioMichaliszyn13}
and~\cite{Montanari+14a}. Specifically, \ehs{} can express properties
of complex processes consisting of many stages, even if the processes
are repeating or overlapping.  Regular expressions allow to express
further properties not explored here. 


We showed that the model checking for the $BDE$ fragment of \ehs{} is
decidable and \PSPACE{}-complete, and that the model checking problem
for the $A\bei{}L$ fragment of the logic is decidable. While the
complexity is the same as that for the EHS logic
in~\cite{LomuscioMichaliszyn14c}, \ehs{} is considerably more
expressive.

Further ahead we intend to study more expressive fragments of \ehs{}.
We believe that the technique presented here can be extended to
backward modalities, such as $\naw{\aei}$, $\naw{\dei}$, $\naw{\eei}$,
$\naw{\lei}$ and $\naw{\nei}$.  However, more investigations are
required, since in the case of backward modalities one cannot simply
disregard the histories.

A further open problem is the decidability of any fragment involving
the modality $O$.  In a sense, $O$ is the hardest case of all
operators. Indeed, is known that the satisfiability for the $O$
fragment of HS is undecidable \cite{Bresolin+11a}.  Since $O$ can be
expressed using $\bar{B}$ and $E$ \cite{M11}, we cannot show the
decidability of the join of the fragments studied in this paper
($AB\bar{B}DELN$) without proving it for $O$.

Finally, we are interesting in implementing an efficient model
checking toolkit for \ehsre{} specifications. We intend to develop
more efficient algorithms on symbolic representations and a suitable
predicate abstraction technique for \ehsre{}.

\medskip

\noindent\textbf{Acknowledgments.}
The authors would like to thank Angelo Montanari whose comments on
\cite{LomuscioMichaliszyn14c} lead to the present investigation. \\
The second author was generously supported by Polish National Science Center based on the decision
number DEC-2011/03/N/ST6/00415. This research was supported by the EPSRC under grant EP/I00520X.

\vspace{-2em}

\bibliographystyle{splncs03}
\bibliography{all,bib}

\appendix

\section{Sketch of the Proof of Theorem \ref{t:reductions}}
Roughly speaking, functions $f$ and $f'$ just move the regular expressions from the labelling to the formula and the other way round.
Function $f$ is such that $f(\lambda, \varphi) = (\lambda', \varphi')$, where $\lambda'(g)=g$ for all the states $s$ and $\varphi'$ is the result of replacing each propositional variable $q$ in $\varphi$ by $\sum_{g \in \lambda(q)} g$.
Function $f'$ is such that $f'(\lambda', \varphi') = (\lambda, \varphi)$, where for each regular expression $r$ in $\varphi'$, we replace $r$ by an unique propositional variable $q^r$ and we put $\lambda(q^r) = \lambda' \circ r$. 
It is readily verifiable that both functions are as required.

\section{Proof of Theorem \ref{t:bde}}

\begin{proof}
The lower bound follows from the lower bound for the endpoint-based
variant of \isrl{} that was shown in \cite{LomuscioMichaliszyn13} for
the same syntax.

For the upper bound, we consider an alternating algorithm \cite{CKS81} working in
polynomial time.  Since \textsc{APTime}=\textsc{PSpace}, the theorem
follows.  Algorithm~\ref{a:modelcheckingbde} reports the procedure
\textsc{ver-BDE} that solves the model checking problem.  Its
complexity follows from the fact that each existentially or
universally selected interval has the size bounded by the size of the
initial interval.
\qed
\end{proof}

\begin{algorithm}
\caption{The model checking procedure for \ehsd{}.}
\label{a:modelcheckingbde}
\begin{algorithmic}[1]
\setlength{\itemsep}{0pt}
\Procedure{ver-BDE}{${M}$, $I$, $\varphi$}
 \If{$\varphi = p$}
   {\textbf{return {$g(I) \in \Lang{\lambda(p)}$}} }
 \EndIf
 \If{$\varphi = pi$}
   {\textbf{return {$pi(I)$} }}
 \EndIf
 \If{$\varphi = \neg \varphi'$}
   {\textbf{return} $\neg$\Call{ver-BDE}{${M}$, $I$, $\varphi'$}}
 \EndIf
 \If{$\varphi = \varphi_1 \land \varphi_2$}
    {\scalebox{0.92}[1.0]{\textbf{return} \Call{ver-BDE}{${M}$, $I$, $\varphi_1$}$\land$\Call{ver-BDE}{${M}$, $I$, $\varphi_2$}}}
 \EndIf
 \If{$\varphi = K_i \varphi'$ where $i \in A$}
  \State{universally select $J$ such that $J \sim_i I$}
  \State {\textbf{return} \Call{ver-BDE}{${M}$, $J$, $\varphi'$}}
 \EndIf
 \If{$\varphi = C_G \varphi'$ where $G \subseteq A$}
  \State{universally select $J$ such that  $J \sim_G I$}
   \State \textbf{return} \Call{ver-BDE}{${M}$, $J$, $\varphi'$}
 \EndIf
 \If{$\varphi = X \varphi'$ where $X \in \set{\naw{B}, \naw{D}, \naw{E}}$}
  \State{existentially select $J$ such that $I R_X J$}
   \State  {\textbf{return} \Call{ver-BDE}{${M}$, $J$, $\varphi'$}}
 \EndIf
\EndProcedure
\end{algorithmic}
\end{algorithm}

\section{Proof of Theorem \ref{t:abldec}}
Observe that $\naw{L}$ can be defined in terms of $\naw{A}$: for any $\varphi$, $\naw{L} \varphi \equiv \naw{A}(\neg
pi \wedge \naw{A} \varphi)$.  Given this, in what follows we assume
that the formulas do not contain $\naw{L}$. We now define some
auxiliary notions.

For convenience, for each modality $X$  of \ehsr{}, we define a relation $R_X$ as follows: $R_{\naw{A}}=R_A$,
$R_{\naw{\bei}}=R_{\bei}$, $R_{K_i} = \sim_i$ and $R_{C_G}=\sim_{G}$.

\begin{theorem}\label{t:boundeddec}
Model checking ISRL under bounded semantics against \ehsr{}
specifications is decidable.
\end{theorem}

\begin{algorithm}
\caption{The model checking procedure for \ehsr{}.}
\label{a:modelchecking}
\begin{algorithmic}[1]
\setlength{\itemsep}{0pt}
\Procedure{verify}{${M}$, $I$, $\varphi$}
 \If{$\varphi = p$}
    {\textbf{return {$I \in \Lang{\lambda(p)}$}} }
 \EndIf
 \If{$\varphi = pi$}
    {\textbf{return {$pi(I)$} }}
 \EndIf
 \If{$\varphi = \neg \varphi'$}
    {\textbf{return} $\neg$\Call{verify}{${M}$, $I$, $\varphi'$}}
 \EndIf
 \If{$\varphi = \varphi_1 \land \varphi_2$}
    {\textbf{return} \Call{verify}{${M}$, $I$, $\varphi_1$} $\land$ \Call{verify}{${M}$, $I$, $\varphi_2$}}
 \EndIf
 \If{$\varphi = K_i \varphi'$ where $i \in A$}
  \ForAll{$J$ s.t. $I \sim_i J$}
   \If{$\neg$\Call{verify}{${M}$, $J$, $\varphi'$}}
     {\textbf{return} false}
   \EndIf
  \EndFor
  \State{\textbf{return} true}
 \EndIf
 \If{$\varphi = C_G \varphi'$ where $G \subseteq A$}
  \ForAll{$J$ s.t. $I \sim_G J$}
   \If{$\neg$\Call{verify}{${M}$, $J$, $\varphi'$}}
     {\textbf{return} false}
   \EndIf
  \EndFor
  \State{\textbf{return} true}
 \EndIf
 \If{$\varphi = X \varphi'$ where $X \in \set{\naw{A}, \naw{\bei}}$}
  \ForAll{$J$ s.t. $I R_X J$ and $|J| \leq f(\varphi)+|I|$}
   \If{\Call{verify}{${M}$, $J$, $\varphi'$}}
     {\textbf{return} true}
   \EndIf
  \EndFor
  \State{\textbf{return} false}
 \EndIf
\EndProcedure
\end{algorithmic}
\end{algorithm}

\begin{proof}The procedure $\Call{Verify}$ given in Algorithm
\ref{a:modelchecking} solves the model checking problem.  Clearly, it
always terminates and its computation time is non-elementary.
\qed
\end{proof}

Our crucial theorem says that the bounded semantics is basically the same as the unbounded one.

\begin{theorem}\label{t:equivalence}
Given an \ehsr{} formula $\varphi$, a model $M$, and an interval $I$,
$M, I \models \varphi$ if and only if $M, I \modelsb \varphi$.
\end{theorem}

\begin{proof}
Consider a model $M=\Me$.  For each $p \in \Var$ we denote by
$\A^p$ the minimal deterministic finite state automaton
\cite{HopcroftUllman79} recognising the language $\Lang{\lambda(p)}$. 
By $\A_w(p)$, where $p \in \Var$, we denote the state of $\A^p$ after
reading a word $w$; in the following, we treat $\A_w$ as a function from $\Var$ to automata states.

\begin{definition}[Modal Context Tree]
Given a model $M$, the \emph{modal context tree} of an interval $I$
w.r.t. an \ehsr{} formula $\varphi$, denoted by $MCT_I^\varphi$, is the
minimal unranked tree with labelled nodes and edges defined recursively as
follows.
\begin{itemize}
\item The root of the tree is labelled by the tuple $g(\first(I)), g(\last(I)), pi(I), \A_I$.
\item For each top-level sub-formula $X \psi$ of $\varphi$ and each
  interval $I'$ such that $I R_X I'$, the root of $MCT_I^\varphi$ has
  an $X\psi$-successor $MCT_{I'}^\psi$
  ($X$ indicates the labelling of an edge).
\end{itemize}
\end{definition}

In other words $MCT_I^\varphi$ contains sufficient information about
all the intervals that need to be considered to determine the value of
$\varphi$ in $I$ as well as the states of the automata after reading
$I$.

\begin{example}\label{e:mct}
Consider the \isrl{} $IS_{ex}$ from Example \ref{e:re1}, the
formula $\varphi = K_0 pi \wedge \neg \naw{A} p$, and an interval
$I=g_1$.

To build the modal context tree, we use the automaton for $\lambda(p)$
presented in Figure \ref{f:automaton}.

\begin{figure}
\centering
\begin{minipage}{.25\textwidth}
  \centering
  \begin{tikzpicture}[shorten >=1pt,->, scale=1]
  \tikzstyle{vertex}=[circle,draw=black,minimum size=16pt,inner sep=0pt]

\begin{scope}[shift={(1,0)}]
  \foreach \var/\x/\y/\name in {z_1/0.8/1/s1, z_2/2/1/s2, z_3/3.2/1/s3, z_\bot/2/-0.4/s4}
    \node[vertex] (G-\name) at (\x,\y) {$\var$};

  \draw (G-s1) to node[above] {$g_1$} (G-s2);    
  \draw (G-s1) to node[below,label={[label distance=-0.2cm]45:$g_2, g_3$}] {} (G-s4);    
  \draw (G-s2) to node[above] {$g_3$} (G-s3);    
  \draw[loop above] (G-s2) to node[above] {$g_1, g_2$} (G-s2);    
  \draw (G-s3) to node[below] {$*$} (G-s4);    
  \draw[loop below] (G-s4) to node[below] {$*$} (G-s4);    
  \draw (0.5, 1.5) to (G-s1);
    
\end{scope}
\end{tikzpicture}
  \captionof{figure}{A minimal automaton for $g_1(g_1+g_2)^*g_3$. $z_3$ is the only accepting state.}
  \label{f:automaton}
\end{minipage}%
\begin{minipage}{.05\textwidth}
~
\end{minipage}%
\begin{minipage}{.7\textwidth}
  \centering
  \begin{tikzpicture}[shorten >=1pt,->, scale=1]
  \tikzstyle{vertex}=[circle,draw=black,minimum size=16pt,inner sep=0pt]

\begin{scope}[shift={(0,0)}]

  \node[vertex,label=left:{$g_1,g_1, \top,\set{(p,z_2)}$}] (G-s1) at (0,4.6) {};

  \foreach \var/\vart/\x/\y/\name in {
{g_1,g_1, \top}/{\set{(p, z_2)}}/-4/3/s2,
{g_2,g_2, \top}/{\set{(p, z_\bot)}}/-2.2/3/s3,
{g_3,g_3, \top}/{\set{(p, z_\bot)}}/-0.4/3/s4,
{g_1,g_1, \top}/{\set{(p, z_2)}}/1/3/s5,
{g_1,g_3, \bot}/{\set{(p, z_3)}}/3/3/s7}
    \node[vertex,label={[align=center]below:{$\var$\\$\vart$}}] (G-\name) at (\x,\y) {};

  \node (G-s6) at (2,3) {\dots};

  \foreach \from/\to/\l in {s1/s2/{K_0 pi}\;\;\;,s1/s3/{K_0 pi}\;,s1/s4/{K_0 pi}}
    \draw (G-\from) edge node[left,pos=0.7] {$\l$} (G-\to);

  \foreach \from/\to/\l in {s1/s5/\naw{A}p,
  s1/s7/\;\;\;\naw{A}p}
    \draw (G-\from) edge node[right,pos=0.7] {$\l$} (G-\to);
\end{scope}

\end{tikzpicture}
  \captionof{figure}{$MCT_I^\varphi$ from Example \ref{e:mct}.
The omitted $\naw{A} p$ successors are labelled by: 
$g_1$, $g_2$, $\bot$, $\set{(p,z_2)}$;
$g_1$, $g_1$, $\bot$, $\set{(p,z_2)}$;
$g_1$, $g_1$, $\bot$, $\set{(p,z_\bot)}$;
$g_1$, $g_2$, $\bot$, $\set{(p,z_\bot)}$;
$g_1$, $g_2$, $\bot$, $\set{(p,z_\bot)}$.
}
  \label{f:examplet}
\end{minipage}
\end{figure}

The top level sub-formulas of $\varphi$ are $K_1 pi$ and $\naw{A} p$.
$MCT_I^\varphi$ (Figure \ref{f:examplet}) represents $I$. Notice that
there are infinitely many $R_A$ successors of $I$, but $MCT_I^\varphi$
needs only 7 $\naw{A} p$-successors.  For example, the successor
labelled by $g_1, g_2, \bot, \set{(p, z_2)}$ represents all the
intervals $I$ such that $g(I)$ is of the form $g_1 (g_1+g_2)^*$.
\end{example}

We now show that the number of modal context trees for a given formula
is bounded.  We will use this later as a kind of pumping argument and
show that is an interval is long enough, then some of its prefixes
have the same modal context tree.

 \begin{lemma}\label{l:thelemmaone}
 Given a model $M$ and a formula $\varphi$,
 $|\set{ MCT_I^\varphi \mid I \text{ is an
     interval in } M }| \allowbreak <f^{IS}(\varphi)$.
 \end{lemma}

\begin{proof} We show the lemma ~by induction on $\varphi$.
Clearly, if a formula has no modalities, then $\set{ MCT_I^\varphi \mid I \text{ is an interval in } M }$ contains trees with only one node,
that can be labelled with $2|G|^2 \prod_{q \in \Var} 2^{\card{\lambda(q)}}$ different labels.
 
Consider a formula $\varphi$
with the top-level sub-formulas $X_1 \varphi_1$, \dots, $X_k \varphi_k$.
Each tree for $\varphi$ consists of one of $2|G|^2 \prod_{q \in \Var} 2^{\card{\lambda(q)}}$ possible roots and,
for each $i$, any subset of subtrees for $\varphi_i$.  Therefore,  $ |\set{ MCT_I^\varphi \mid I \text{ is an
     interval in } M }|< 2|G|^2 \prod_{q \in \Var} 2^{\card{\lambda(q)}} 2^{f^{IS}(\varphi_1)} \dots
2^{f^{IS}(\varphi_k)} = f^{IS}(\varphi)$.
\qed
\end{proof}

We show that the modal context tree does not depend on the histories.

\begin{lemma}\label{l:thelemmatwopr}
Consider a model $M=\Me$ and a formula $\varphi$. If $I$ and $I'$ are
intervals such that $g(I)=g(I')$, then $MCT_I^\varphi=MCT_{I'}^\varphi$.
\end{lemma}
 \begin{proof}
 We show this by induction.
 
The roots of 
 $MCT_I^\varphi$ and $MCT_{I'}^\varphi$ 
have the same labels, since 
$g(\first(I))=g(\first(I'))$,
$g(\last(I))=g(\last(I'))$,
$pi(I)=pi(I')$ and the labelling is defined on $g(I)$.

Consider a $\naw{X} \varphi'$-successor $T$ of the root of $MCT_I^\varphi$, where 
$\naw{X} \varphi'$ is a top-level sub-formula of $\varphi$ and $X \in \set{A, \bei{}, N}$.
There is an interval $J$ such that $I R_X J$ and $MCT_J^{\varphi'} = T$.
So there exists a $J'$ such that $I' R_X J'$ and $g(J)=g(J')$, because $X$ is a ``forward modality'' so the $R_X$ successors of $I'$ do not depend on the history. 
By the inductive hypothesis, $MCT_J^{\varphi'} = MCT_{J'}^{\varphi'}$,
and therefore the roots of $MCT_I^{\varphi}$ and $MCT_{I'}{\varphi}$ have the same $\naw{X} \varphi'$ successors.

As for the $X \varphi'$ successors where $X$ is an epistemic modality, it is enough to observe that $I R_X I'$, and therefore $I$ and $I'$ are related to the same intervals by the equivalence relation $R_X$. 
The lemma follows.
\qed
\end{proof}

Now we argue that if two intervals have the same modal context tree w.r.t. $\varphi$, then either both satisfy $\varphi$ or none of them.

 \begin{lemma}\label{l:thelemmatwo}
Consider a model $M=\Me$ and a formula $\varphi$. If $I$ and $I'$ are
intervals such that $MCT_I^\varphi=MCT_{I'}^\varphi$, then $M, I
\models \varphi$ if and only if $M, I' \models \varphi$.
 \end{lemma}

\begin{proof}
We show it by induction on $\varphi$.  


\noindent\textbf{Case 1.} $\varphi=p$ for some variable $p$.   The
root of the $MCT_{I}^\varphi$ is labelled by the state of an automaton
corresponding to $\lambda(p)$ after reading $I$, and the root of the
$MCT_{I'}^\varphi$ is labelled by the state of an automaton
corresponding to $\lambda(p)$ after reading $I'$.  Since the two trees
are equal, the automaton is in the same state in both cases, either
accepting or rejecting, and therefore ${M}, I \models p$ if and only
if ${M}, I' \models p$.


\noindent\textbf{Case 2.} $\varphi=pi$. 
The root of the $MCT_{I}^\varphi$ is
labelled by $pi(I)$, and so is the root of $MCT_{I'}^\varphi$, and
therefore $pi(I)=pi(I')$.


\noindent\textbf{Case 3.} $\varphi=\neg \varphi'$. By the inductive
assumptions, ${M}, I \models \varphi'$ if and only if ${M}, I' \models
\varphi'$, so ${M}, I \models \varphi$ if and only if ${M}, I' \models
\varphi$.


\noindent\textbf{Case 4.} $\varphi=\varphi_1 \land \varphi_2$. By
the induction assumption, ${M}, I \models \varphi_1$ if and only if
${M}, I' \models \varphi_1$ and ${M}, I \models \varphi_2$ if and only
if ${M}, I' \models \varphi_2$, so ${M}, I \models \varphi$ if and
only if ${M}, I' \models \varphi$.


\noindent\textbf{Case 5.} $\varphi=K_i \varphi'$. Assume that $M, I
\models \varphi$.  Consider any interval $J'$ such that $I' \sim_i
J'$.  By definition, in the tree $MCT_{I'}^\varphi$ the subtree
$MCT_{J'}^{\varphi'}$ is a $K_i \varphi'$-successor of the root.  It
follows that in the tree $MCT_{I}^\varphi$(=$MCT_{I'}^\varphi$), $MCT_{J'}^{\varphi'}$ is
a $K_i\varphi'$-successor of the root.  Let $J$ be such that $I \sim_i
J$ and $MCT_{J'}^{\varphi'} = MCT_{J}^{\varphi'}$.  Clearly, since $M, I
\models \varphi$, $M, J \models \varphi'$.  By the inductive
assumptions, $M, J' \models \varphi'$.  Therefore $M, I' \models
\varphi$.


\noindent\textbf{Case 6.} $\varphi=C_G \varphi'$. Assume that $M, I
\models \varphi$ and $J'$ is such that $I' \sim_G J'$.  Again, in
$MCT_{I'}^\varphi$ the subtree $MCT_{J'}^{\varphi'}$ is a $C_G
\varphi'$-successor of the root.  It follows that in the tree
$MCT_{I}^\varphi$, $MCT_{J'}^{\varphi'}$ is a $C_G\varphi'$-successor
of the root.  Let $J$ be such that $I \sim_G J$ and
$MCT_{J'}^{\varphi'} = MCT_{J}^{\varphi'}$, then $M, J \models
\varphi'$, and by the inductive assumptions, $M, J' \models \varphi'$.
Therefore $M, I' \models \varphi$.


\noindent\textbf{Case 7.} $\varphi=\naw{A} \varphi'$.
We have ${M}, I \models \naw{A}\varphi'$ if and only if there is an
interval $J$ starting in $last(I)$ satisfying $\varphi'$.  Since
$g(last(I)) = g(last(I'))$, the intervals starting from $last(I)$ and
$last(I')$ are the same (modulo histories), and therefore there exists
an interval $J'$ starting in $last(I')$ such that $g(J)=g(J')$.  By
Lemma \ref{l:thelemmatwopr}, it follows that $MCT_J^{\varphi'} =
MCT_{J'}^{\varphi'}$.


\noindent\textbf{Case 8.} $\varphi=\naw{\bei} \varphi'$. Assume that
there is an interval $J$ such that $I R_{\bei{}} J$ and $M, J \models
\varphi'$.  Then, $MCT^{\varphi'}_J$ is an $\naw{\bei} \varphi'$
successor of the root in $MCT^{\varphi}_I$, and so in
$MCT^{\varphi}_{I'}$.  So there is an interval $J'$ such that $I'
R_{\bei{}} J'$ and $MCT^{\varphi'}_J = MCT^{\varphi'}_{J'}$.  By the
inductive hypothesis, $M, J' \models \varphi'$ and therefore $M, I'
\models \varphi$.


\noindent\textbf{Case 9.} $\varphi=\naw{N} \varphi'$.
This can be shown similarly to Case 7 for $\naw{A}\varphi'$.
\qed
\end{proof}

As we remarked earlier, if an interval $I$ is long enough, then $I$ has
two prefixes with the same modal context tree w.r.t. a formula
$\varphi$.  Intuitively speaking, we would like to replace the longer
prefix by the shorter one, thereby obtaining an interval $I'$, and
show that the modal context trees of $I$ and $I'$ are the same. By the
above lemma, it would follow that they both satisfy the given formula.
What remains to be proved is that if we have two prefixes with the
same modal context tree, and we append the same interval to both, the
results will also have the same modal context tree.

We use the following terminology.  A \emph{partial state} is a
sequence of states $g_1 \dots g_k$ such that for all $i < k$,
we have $t^G(g_{i}, g_{i+1})$.  Each state of the model is a partial
state; but partial states are not required to start at $g_0$.  A \emph{partial interval} is a sequence $s_1\dots
s_k$ of partial states such that for each $i < k$ we have that
$s_{i+1}=s_ig_i$ for some partial state $g_i$.  A partial interval
$I=s_1\dots s_k$ is \emph{clear} if $s_1=g$ for some partial state
$g$.  We extend the functions $\first$, $\last$, and $g$ and the other
notions to partial intervals in the obvious way.

We define the operation of adding context to partial intervals as
follows.  Given a partial interval $I$ and a clear partial interval
$I'=s_1 \dots s_{k}$ where $t^G(g(\last(I)), g(\first(I')))$, by $I
\oplus I'$ we denote the partial interval $I \bar{s}_1 \dots
\bar{s}_k$ such that for each $i$ we have that $\bar{s}_i = last(I)
s_i$.  So $\oplus$ joins two intervals in a way that accounts for the
history of the partial states.  Clearly, $I \oplus I'$ is an interval
if and only if $I$ is an interval.  We also define the operation
$\circ$ such that $I \circ I' = \bar{s}_1 \dots \bar{s_k}$, i.e., it
only returns the adjusted partial states of $I'$.

 \begin{lemma}\label{l:thelemmathree}
 Consider a model $M$, a formula $\varphi$, two intervals $I, I'$, and a
 partial interval $J$.  If $MCT_I^\varphi=MCT_{I'}^\varphi$, and
 $t^G(g(\last(I)), g(\first(J)))$, then $MCT_{I\oplus
   J}^\varphi=MCT_{I' \oplus J}^\varphi$.
 \end{lemma}

\begin{proof}
Consider a formula $\varphi$, a model $M$, two intervals $I$, $I'$ and a
partial state $s=g$ such that $t^G(g(\last(I)), g)$.  We show that if
$MCT_{I}^\varphi = MCT_{I'}^\varphi$, then $MCT_{I \circ s}^\varphi =
MCT_{I' \circ s}^\varphi$.  The consideration above can be used to
prove the lemma by induction.

Assume that the root of $MCT_{I}^\varphi$ is labelled by $f$, $l$,
$pi$, $\A_I$.  Then the roots of both $MCT_{I\circ s}^\varphi$ and
$MCT_{I'\circ s}^\varphi$ are labelled by $f$, $g$, $\bot$, $\A$,
where for each $p\in\Var$ we put $\A{}(p)$ equal to the state that the
automaton for $p$ reaches from $\A_I(p)$ after reading $g$.

Assume that $X_1\varphi_1$, \dots, $X_k \varphi_k$ are the top-level
sub-formulas of $\varphi$ and $i\in \set{1, \dots, k}$ (if there are no
such formulas, then the result follows directly).  We show that for
each $i$, the roots of $MCT_{Is}^\varphi$ and $ MCT_{I's}^\varphi$
have the same $X_i\varphi_i$-successors.


\noindent\textbf{Case 1.} $X_i$ is an epistemic modality.
Consider any interval $J$ such that $I \oplus s R_{X_i} J$.
Let $J=J' \oplus s'$. 
By the definition, $J' R_{X_i}  I$ and $s R_{X_i} s'$.
By the former, 
we have that $MCT_{J'}^{\varphi_i}$ is an $X_i\varphi_i$-successor of the root in $MCT_{I \oplus s}^\varphi$, and so $MCT_{J'}^{\varphi_i}$ is an $X_i\varphi_i$-successor of the root in $MCT_{I'}^\varphi$.
So there is $J'' R_{X_i} I'$ such that $MCT_{J'}^{\varphi_i}=MCT_{J''}^{\varphi_i}$.
Therefore, $J'' \oplus s' R_{X_i} I' \oplus s$, and thus $MCT_{J}^{\varphi_i}$ is the $X_i\varphi_i$-successors of the root of $MCT_{I\oplus s}^\varphi$.


\noindent\textbf{Case 2.} $X_i = \naw{A}$.  Consider any interval $J$
such that $I \oplus s R_A J$.  Then there is a clear partial interval
$\bar{J}$ such that $J = I \circ \bar{J}$.  Let $J' = I' \circ
\bar{J}$.  It holds that $I' \circ s R_A J'$.  By Lemma
\ref{l:thelemmatwopr}, we have
$MCT_{J}^{\varphi_i}=MCT_{J'}^{\varphi_i}$.

Therefore, the $\naw{A}\varphi_i$-successors of the root in
$MCT_{I\oplus s}^\varphi$ are also $\naw{A}\varphi_i$-successors of
the root in $MCT_{I'\oplus s}^\varphi$.  The other direction is
similar.


\noindent\textbf{Case 3.} $X_i = \naw{\bei}$.  Consider any interval
$J$ such that $I \oplus s R_{\bei} J$.  Then, there is a clear partial
interval $\bar{J}$ such that $J = (I\oplus s) \oplus \bar{J}$.

Let $J' = (I' \oplus s) \oplus \bar{J}$.  It holds that $I' \oplus s
R_{\bei} J'$.  By Lemma \ref{l:thelemmatwopr}, we have
$MCT_{J}^{\varphi_i}=MCT_{J'}^{\varphi_i}$.

Again, we conclude that the $\naw{\bei{}}\varphi_i$-successors of the root in $MCT_{I \oplus s}^\varphi$ are the same as $\naw{\bei}\varphi_i$-successors of the root in $MCT_{I \oplus s}^\varphi$.


\noindent\textbf{Case 4.}  $X_i = \naw{N}$.
The proof is similar to the one of Case 2 for $X_i = \naw{A}$.
\qed
\end{proof}

By exploiting the Lemma above, we can now give the main result of
this section.

The proof of Theorem \ref{t:equivalence}
 is by induction on the structure of $\varphi$.

The cases for $\varphi=p$, $\varphi=pi$, $\varphi=\neg \varphi'$,
$\varphi=\varphi_1 \land \varphi_2$, $\varphi=K_i \varphi'$, and
$\varphi=C_G \varphi'$ for some sub-formulas $\varphi', \varphi_1,
\varphi_2$, follow from the fact that the semantic rules are the same
in both semantics.

Assume that $\varphi=X \varphi'$ for some $\varphi'$, and $X \in
\naw{A}, \naw{\bei}, \naw{N}$. 
If $M, I \modelsb \varphi$, then there is an
interval $I'$ of bounded size such that $M, I' \modelsb \varphi'$
and $I R_X I'$.  By the induction hypothesis, $M, I' \models
\varphi'$ and therefore $M, I \models \varphi$.  

If $M, I \models \varphi$, then there is an interval $I'$ such that
$M, I' \models \varphi'$ and $I R_X I'$.  Let $I'$ be the shortest
possible interval with this property.  We show that $|I'|\leq
|I|+f^{IS}(\varphi)$.

Let $I'=s_1 \dots s_t$ and $I'_k$ denote the prefix $s_1 \dots s_k$ of
$I'$.  Assume that $|I'|>|I|+f^{IS}(\varphi')$.  By
Lemma~\ref{l:thelemmaone} there are two prefixes $I'_k$, $I'_l$ such
that $|I|<k<l$ and $MCT_{I'_k}^{\varphi'}=MCT_{I'_l}^{\varphi'}$.

Let $J$ be a clear partial interval such that $I' = I'_l \oplus J$.
By Lemma~\ref{l:thelemmathree}, we have that $MCT_{I'_k \oplus
  J}^{\varphi'}=MCT_{I'_l \oplus J}^{\varphi'}$ Clearly, $|I'_k \oplus
J|<|I'|$ and, by Lemma \ref{l:thelemmatwo}, $M, I'_k \oplus J \models
\varphi'$.  Since $k>|I|$, it follows that $I R_X I'_k \oplus J$ (the
condition $k>|I|$ is only required for $\naw{\bei}$ since $J$ has to
contain $I$ as a prefix).  But we assumed that $I'$ was the shortest
interval; so this is a contradiction. It follows that $|I'|\leq
|I|+f^{IS}(\varphi)$.
\qed
\end{proof}

Finally, the proof of Theorem \ref{t:abldec} goes as follows.
By Theorem~\ref{t:equivalence}, the bounded semantics and the
unbounded semantics are equivalent.  By Theorem~\ref{t:boundeddec},
model checking the $A\bar{B}LN$ fragment of \ehs{} with bounded
semantics is decidable.  Therefore, model checking the $A\bar{B}LN$
fragment of \ehs{} with unbounded semantics is also decidable.
Indeed, the procedure \textsc{Verify} given in
Algorithm~\ref{a:modelchecking} solves the problem.

\end{document}